\documentclass[11pt]{article}
\usepackage{graphicx}
\usepackage[margin=1.1in]{geometry}
\usepackage{bm}
\usepackage{subfig}
\usepackage{natbib}
\usepackage{amsmath,amssymb,amsfonts,amsthm}

\newcommand{\N}{\ensuremath{\mathbb{N}}}
\newcommand{\R}{\ensuremath{\mathbb{R}}}
\newcommand{\PP}{\ensuremath{\mathbb{P}}}
\newcommand{\E}{\ensuremath{\mathbb{E}}}

\newcommand{\event}{\ensuremath{F}}
\newtheorem{theorem}{Theorem}
\newtheorem{lemma}{Lemma}

\newtheorem{remark}{Remark}
\newtheorem{example}{Example}
\newtheorem{algorithm}{Algorithm}
\newtheorem{condition}{Condition}
\newtheorem{definition}{Definition}

\title{A Framework for Monte Carlo based Multiple Testing}
\author{Axel Gandy and Georg Hahn
  \\Department of Mathematics, Imperial College London}
\date{}

\begin{document}
\maketitle

\begin{abstract}
We are concerned with a situation in which we would like to test multiple hypotheses
with tests whose p-values cannot be computed explicitly but can be approximated using Monte Carlo simulation.
This scenario occurs widely in practice.
We are interested in obtaining the same rejections and non-rejections
as the ones obtained if the p-values for all hypotheses had been available.
The present article introduces a framework for this scenario
by providing a generic algorithm for a general multiple testing procedure.
We establish conditions which guarantee that the rejections and
non-rejections obtained through Monte Carlo simulations are identical to
the ones obtained with the p-values.
Our framework is applicable to
a general class of step-up and step-down procedures
which includes many established multiple testing corrections such as the ones of Bonferroni,
Holm, Sidak, Hochberg or Benjamini-Hochberg.
Moreover, we show how to use our framework to improve
algorithms available in the literature
in such a way as to yield theoretical guarantees on their results.
These modifications can easily be implemented in practice and
lead to a particular way of reporting multiple testing results as three sets
together with an error bound on their correctness,
demonstrated exemplarily using a real biological dataset.
\end{abstract}


\section{Introduction}
\label{section_introduction}
We would like to test $m$ hypotheses $H_{01}, \ldots, H_{0m}$
for statistical significance using a multiple testing procedure
given by a mapping
\begin{align}
h: [0,1]^m \times [0,1] \rightarrow \mathcal{P}(\{ 1,\ldots,m \})
\label{label_h}
\end{align}
which takes a vector of $m$ p-values $p \in [0,1]^m$ and a
threshold $\alpha \in [0,1]$ and returns the set of indices of
hypotheses to be rejected, where $\mathcal{P}$ denotes the power set.
As we consider control of both
the familywise error or the false discovery rate,
this procedure could, for instance,
be the \cite{Bonferroni1936} correction,
the \cite{Sidak1967} correction
or the procedures of \cite{Holm1979},
\cite{Hochberg1988} or \cite{Benjamini1995CFD}.

We assume that the p-values $p^\ast=(p_1^\ast,\ldots,p_m^\ast)$
of the underlying tests cannot be computed explicitly.
Moreover, the threshold $\alpha^\ast$ at which we would
like to test
may depend on the unknown $p^\ast$ and thus may be unknown itself.
Our aim is to compute $h(p^\ast,\alpha^\ast)$.

We assume that we can approximate $p^\ast$ and $\alpha^\ast$
through Monte Carlo simulations from a well defined null distribution.
Such Monte Carlo simulations are assumed to be carried out separately for each hypothesis.
We define \textit{Monte Carlo based multiple testing}
to be the evaluation of multiple hypotheses using a multiple testing correction
applied to $p^\ast$ and $\alpha^\ast$ approximated via Monte Carlo simulations.
Assuming a test statistic is available to test all hypotheses,
the Monte Carlo simulations considered in the present article
correspond to the simulation of independent datasets under the null
and to the evaluation of the test statistics on the simulated data.
A Monte Carlo p-value can then be computed as a proportion
of simulated test statistics exceeding the test statistic evaluated on the observed data
\citep{Sandve2011}.

Monte Carlo simulations are commonly carried out by
resampling the data in case of bootstrap tests or by
generating permutations when using permutation tests.
This scenario occurs widely in practical situations
\citep{Chen2013,Nusinow2012,Rahmatallah2012}.

More precisely, we assume that we can compute sequences of
nested confidence intervals $I_1^i \supseteq I_2^i \supseteq \cdots$,
$i = 1,\ldots,m$,
for $p_1^\ast,\ldots,p_m^\ast$
and $I_1^{m+1} \supseteq I_2^{m+1} \supseteq \cdots$
for $\alpha^\ast$.
For some of our results,
we require that these intervals $\left( I_n^i \right)_{i,n}$
have a positive joint coverage probability and that their lengths uniformly go to zero.

Being able to compute $h(p^\ast,\alpha^\ast)$ through Monte Carlo simulation
ensures the repeatability and objectivity of the test result \citep{GandyHahn2014}.
Moreover, all guarantees of $h$,
in particular the type I and type II error control, are valid (up to an error probability)
even when using approximations of $p^\ast$ and $\alpha^\ast$.

Existing methods
\citep{BesagClifford1991,Lin2005,Wieringen2008,GuoPedadda2008,Sandve2011}
return a set of rejected hypotheses which,
at least to a certain extent, is random,
where the randomness is coming from the Monte Carlo simulations
and not from the underlying data.
Consider the following example which will be revisited in Section \ref{section_use_practice}.
\cite{Sandve2011} use their method \texttt{MCFDR}
to classify a genome dataset
of \cite{Pekowska2010} with the aim to test if a
gene modification appears more often in certain gene regions.
Each gene region corresponds to one hypothesis.
\cite{Sandve2011} report $2747$ significant hypotheses out of $3466$ hypotheses
without providing guarantees on the stability of their finding.
Recomputing the decisions on all hypotheses shows
considerable variability:
around $353$ of the $3466$ hypotheses are randomly classified
in the sense that they switch from being rejected
to non-rejected in more than $1\%$ of all cases when repeatedly applying \texttt{MCFDR}.
Conclusions based on the significances of these genes should therefore be questioned.

In the present article we will show how algorithms such as \texttt{MCFDR}
can be modified to
give a guarantee on the stability of their findings
and thus how to reduce the Monte Carlo randomness in their results.
These guarantees are conditional on the data on which the testing is carried out.
Although in principle,
increased stability can be achieved by augmenting the number of Monte Carlo samples,
merely increasing the number of samples
does not provide a guarantee on the decision of each hypothesis.
The guarantees provided in the present article
are effective for any number of samples.

The contribution of the article is threefold.
First, Section \ref{section_framework}
provides a framework for multiple hypothesis testing
under the assumption that p-values are not available
and thus have to be approximated using Monte Carlo methods.
The framework is phrased as a generic algorithm
which, under conditions, computes sub- and supersets of $h(p^\ast,\alpha^\ast)$
that converge to $h(p^\ast,\alpha^\ast)$
(Lemma \ref{lemma_monotonicity}
and Theorem \ref{theorem_convergence}
in Section \ref{subsection_framework_convergenceresults}).
Theoretical bounds on the control of the false discovery rate (fdr) are provided
for the proposed algorithm (Lemma \ref{lemma_fdr} in Section \ref{subsection_framework_convergenceresults})
given it is used in connection with a suitable procedure $h$ controlling the fdr.
The framework also incorporates multiple
testing at a (possibly unknown) corrected testing threshold,
for instance using an estimate of the proportion of true null hypotheses.

Second, we show how to use the framework to modify established algorithms
in such a way as to provide certain proven guarantees on their test results
(Section \ref{section_improving}).

Third, we simplify the condition on the multiple testing procedure
in Section \ref{section_procedures},
yielding an easy-to-check criterion for an arbitrary step-up or step-down procedure
(Section \ref{subsection_procedures_stepupstepdown}).
We then use the simplified criterion to show
that many widely used procedures can be employed in our framework
(Section \ref{subsection_procedures_examples}).

One specific implementation of our generic algorithm
is the \texttt{MMCTest} algorithm of \cite{GandyHahn2014}.
\texttt{MMCTest} uses similar results to the ones in
this article
to prove the correctness of its test result up to a pre-specified error probability.
In contrast to the present article
which presents results for a generic algorithm and a generic multiple
testing procedure,
\texttt{MMCTest} focuses on one specific implementation only
as well as on the two specific multiple testing procedures
of \cite{Bonferroni1936} and \cite{Benjamini1995CFD}.
\cite{GandyHahn2014} do not prove
that a correct test result can also be obtained
through appropiate modifications of existing methods.
Moreover, hypothesis testing at a variable testing threshold
is not possible with \texttt{MMCTest},
theoretical bounds on the control of the false discovery rate are not given,
and \cite{GandyHahn2014}
do not provide a simple criterion to prove whether an arbitrary step-up or step-down
procedure allows one to classify hypotheses without knowledge of the p-values.

In Section \ref{section_use_practice} we pick up our
discussion of the biological dataset of \cite{Pekowska2010}.
We show that our proposed
modifications can easily be implemented in practice, come
at virtually no additional computational cost and lead to
a certain way of reporting multiple testing results as three sets together
with an error bound on their correctness.

The article concludes with a discussion in Section \ref{section_discussion}.
All proofs are included in the Appendix.

Throughout the article, let $|\cdot|$ denote the length of an interval
or the size of a set.
Let $\| \cdot \|$ denote the Euclidean norm.
For an interval $I \subset \R$, let $\min I$ and $\max I$
denote its lower and upper limit, respectively.
For any set $S \subseteq \{ 1,\ldots,m \}$,
where $m \in \N$, let $S^c$ denote the complement of $S$ with respect
to $\{ 1,\ldots,m \}$.
We abbreviate $(x_1,\ldots,x_n)$ by $x_{1:n}$, where $x_{1:0}=\emptyset$.

\section{The framework}
\label{section_framework}
Our framework includes two components, the multiple testing procedure $h$
and a generic algorithm presented in Section \ref{subsection_framework_algorithm}.
Section \ref{subsection_improved_naive} discusses a concrete implementation
of our generic algorithm with the aim to exemplarily improve an established method.
Combining both the testing procedure and the algorithm
yields a framework which, under conditions,
guarantees the correctness of its test result
(Section \ref{subsection_framework_convergenceresults}).

\subsection{The generic algorithm}
\label{subsection_framework_algorithm}
We propose to use the following
generic sequential algorithm to draw samples for each hypothesis.
As the p-values are unknown,
in each iteration $n$, the generic algorithm
computes intervals $I_n^i$ for each $p^\ast_i$, $i \in \{ 1,\ldots,m \}$,
as well as an interval $I^{m+1}_n$ for $\alpha^\ast$.
Usually, these intervals will be confidence intervals, in which case
our algorithm will compute sub- and supersets of $h(p^\ast,\alpha^\ast)$
(Section \ref{subsection_framework_convergenceresults}).

Although the threshold will be, in most cases, a function $\alpha^\ast=g(p^\ast)$
of the p-values $p^\ast$ (see Section \ref{section_improving}),
it is sensible to not restrict the multiple testing procedure to $h(p^\ast)=h(p^\ast,g(p^\ast))$
and to keep a separate interval $I^{m+1}_n$ for $\alpha^\ast$ instead:
Naturally, one could use confidence bounds on $p^\ast$ to obtain a plug-in interval for $\alpha^\ast$
(provided that $g$ is monotonic).
However, Example \ref{example_hoeffding} demonstrates that
for the testing threshold of \cite{PoundsCheng2006},
Hoeffding's inequality \citep{Hoeffding1963}
allows one to construct a tighter confidence interval for $\alpha^\ast$ than the plug-in interval,
thus yielding a faster convergence to $h(p^\ast,\alpha^\ast)$
as well as considerably more decisions on individual hypotheses in a real-data study (Section \ref{section_use_practice}).

The generic algorithm draws Monte Carlo samples in each
iteration $n$, denoted by the observations $O_n$.
These are typically sets of samples drawn for all hypotheses or for a subset of the hypotheses.
The decision for which hypotheses to sample new observations may depend on the history
of observations drawn up to iteration $n-1$.
For instance,
in Section \ref{subsection_improved_naive}
which considers the evaluation of multiple tests based on a test statistic, each observation $O_n$
is a vector of indicators signaling if the test statistic evaluated on the $n^{th}$ sample
drawn for each of the $m$ hypothesis exceeds the observed test statistic.

Each (confidence) interval $I_n^i$
is computed by a function $F_i$, $i \in \{ 1,\ldots,m+1 \}$,
using the current history of observations $O_{1:n}$, $n \in \N$.
For generality, we do not impose that $F_i$ computes any specific type
of confidence interval.
Intersecting the intervals in Algorithm \ref{algorithm_generic}
produces a nested sequence of $\left( I_n^i \right)_{n \in \N}$
for each $i \in \{ 1,\ldots,m+1 \}$.

\begin{algorithm}[Generic algorithm]
\label{algorithm_generic}~
\begin{center}
  \begin{tabular}{ll}
    \multicolumn{2}{l}{$\underline{A}_0 = \emptyset, \overline{A}_0 = \{ 1,\ldots,m \},
    I_0^i = [0,1], i \in \{ 1,\ldots,m \}, I_0^{m+1} = \R.$}\\
    For $n=1,2,\ldots:$
    & \text{Choose which $O_n$ to sample based on $O_{1:n-1}$,}\\
    & \text{Sample $O_n$},\\
    & $I_n^i = F_i(O_{1:n}) \cap I_{n-1}^i, i \in \{ 1,\ldots,m+1 \}$,\\
    & $\overline{A}_n = h((\min I_n^i)_{i \in \{1,\ldots,m\}}, \max I_n^{m+1}),$\\
    & $\underline{A}_n = h((\max I_n^i)_{i \in \{1,\ldots,m\}}, \min I_n^{m+1}).$
\end{tabular}
\end{center}
\end{algorithm}
In each iteration $n$,
Algorithm \ref{algorithm_generic} uses the history of samples observed
up to iteration $n-1$ to determine a new set of observations $O_n$ to be sampled.
The key idea of Algorithm \ref{algorithm_generic} is to apply
the multiple testing procedure $h$ to lower ($\min I_n^i$)
and upper ($\max I_n^i$) confidence limits of the $\left( I_n^i \right)_{n \in \N}$, $i \in \{ 1,\ldots,m \}$.
This yields two sets $\overline{A}_n$ and $\underline{A}_n$.

\subsection{The improved na\"ive method}
\label{subsection_improved_naive}
A widely used method in practice to estimate $h(p^\ast,\alpha^\ast)$
is to draw a constant number of samples $s$ for each hypothesis $H_{0i}$, where $i \in \{ 1,\ldots,m \}$,
then compute a point estimate of each p-value
and classify all hypotheses at a constant threshold $\alpha^\ast$
based on these point estimates
\citep{Nusinow2012,Gusenleitner2012,Rahmatallah2012,
Zhou2013,Li2012a,Cohen2012}.
We will call this the na\"ive method.
The na\"ive method can be applied to any multiple testing procedure $h$.

In the following we present
an improvement of the na\"ive method by stating a concrete implementation
of Algorithm \ref{algorithm_generic}.
As shown in Section \ref{section_improving},
under conditions on $h$,
the sets $\underline{A}_n$ ($\overline{A}_n$)
defined in Algorithm \ref{algorithm_generic}
will be subsets (supersets) of $h(p^\ast,\alpha^\ast)$
in each iteration $n$ of our improved na\"ive method
up to a pre-specified error probability.

One key ingredient of the improved na\"ive method are the confidence
sequences given in \cite{Lai1976}: for independent
$Y_1,Y_2,\ldots \sim \text{Bernoulli}(p)$,
$$\PP(g_n^\beta(S_n) < p < f_n^\beta(S_n) ~\forall n \geq 1) \geq 1-\beta,$$
where $S_n = \sum_{i=1}^n Y_i$ and
$g_n^\beta(x) < f_n^\beta(x)$ are the two distinct \citep{Lai1976} roots of
$(n+1) \binom{n}{x} p^x (1-p)^{n-x} = \beta$ for a given $\beta \in (0,1)$.

In many applications of the na\"ive method, multiple tests are based
on a test statistic and it is possible to sample under the null hypothesis.
Let $X_n^i=1$ if the test statistic evaluated on the $n^{th}$ sample
drawn for hypothesis $H_{0i}$
exceeds the observed test statistic, otherwise $X_n^i=0$.
For our improved na\"ive method, we draw one new sample per hypothesis
in each iteration $n$.

The improved na\"ive method is obtained by defining
\begin{align*}
O_n &= (X_n^1,\ldots,X_n^m),\\
F_i(O) &= \left[ g_{|O|}^\beta \left( \sum_{j=1}^{|O|} O_j^i \right),
f_{|O|}^\beta \left( \sum_{j=1}^{|O|} O_j^i \right) \right],~i=1,\ldots,m,\\
F_{m+1}(O) &= \{ \alpha^\ast \},
\end{align*}
where $O_j^i=X_j^i$ and $|O_{1:n}|=n$.

Although the above method is open-ended,
we usually stop the improved na\"ive method
after a pre-specified total number of iterations $s$.
In this case, solely the two test results in
$\underline{A}_s$ and $\overline{A}_s$
based on the intervals of the last iteration
will be returned as result of the algorithm.

In the improved na\"ive method,
the testing threshold $\alpha^\ast$ is assumed to be constant.
However, the interval $F_{m+1}(O_{1:n})$ for $\alpha^\ast$
is needed if $\alpha^\ast$ depends on $p^\ast$.
For instance, this is the case
for thresholds depending
on an estimate of the proportion of true null hypotheses
which is usually a functional of $p^\ast$.
Using such an estimated threshold potentially results in more significant hypotheses
which is desired in practice.

Starting with the work of
\cite{SchwederSpjotvoll1982},
many authors have investigated estimators of the proportion of true
null hypotheses, such as
\cite{Storey2002},
\cite{Langaas2005},
\cite{PoundsCheng2006},
\cite{Finner2009}
and
\cite{FriguetCauseur2011}.

\subsection{Convergence results}
\label{subsection_framework_convergenceresults}
This section states our main results for which
we need the following monotonicity property:
\begin{definition}
\label{definition_monotonicity}
$h$ is \textit{monotonic} if
$h(p,\alpha) \subseteq h(q,\alpha')$
for $p \geq q$ and $\alpha \leq \alpha'$.
\end{definition}
A multiple testing procedure is thus monotonic if smaller p-values
\citep[as introduced in][]{TamhaneLiu2008}
or a higher testing threshold \citep[see][]{Roth1999} lead to more rejections.
The estimators of the proportion of true null hypotheses listed in the last paragraph
of Section \ref{subsection_framework_algorithm} all depend on the p-values $p^\ast$ only
(and some tuning parameters)
and are monotonically increasing in $p^\ast$.
When combined with a multiple testing procedure they thus preserve the monotonicity in the threshold argument.

Suppose in each iteration $n \in \N$,
each p-value $p_i^\ast$ is contained in its
interval $F_i(O_{1:n})$, $i \in \{ 1,\ldots, m \}$,
and the testing threshold $\alpha^\ast$
is contained in the interval $F_{m+1}(O_{1:n})$,
expressed as the event
$$\event_1 = \left\{ \alpha^\ast \in F_{m+1}(O_{1:n}), p_i^\ast \in F_i(O_{1:n}) ~\forall i \in \{ 1,\ldots,m \}, n \in \N \right\}.$$
The following lemma shows that on the event $\event_1$,
classifying hypotheses based on upper and lower interval bounds allows
Algorithm \ref{algorithm_generic} to compute sub- and supersets of $h(p^\ast,\alpha^\ast)$
for monotonic multiple testing procedures $h$.

\begin{lemma}
\label{lemma_monotonicity}
Let $h$ be a monotonic multiple testing procedure. Then,
\begin{enumerate}
    \item $\underline{A}_n \nearrow$ and $\overline{A}_n \searrow$ as $n \rightarrow \infty$,
    \item $\underline{A}_n \subseteq h(p^\ast,\alpha^\ast) \subseteq \overline{A}_n$ $\forall n \in \N$
    on the event $\event_1$.
\end{enumerate}
\end{lemma}
The first part of Lemma \ref{lemma_monotonicity} is not dependent on
the event $\event_1$. It follows purely from the construction
of Algorithm \ref{algorithm_generic} which computes nested intervals
for each $p_i^\ast$, $i \in \{ 1,\ldots, m \}$.
The second part of Lemma \ref{lemma_monotonicity} shows that on $\event_1$, in any iteration $n$,
all the hypotheses in the set
$\underline{A}_n$ ($\overline{A}_n^c$) can already be classified as being
rejected (non-rejected).

Additional properties of Algorithm \ref{algorithm_generic}
can be derived for any monotonic multiple testing procedure $h$ and choice of
$p^\ast$, $\alpha^\ast$ which satisfy the following condition.

\begin{condition}
\label{condition_h}
\begin{enumerate}
  \item Let $p,q \in [0,1]^m$ and $\alpha \in \R$.
If $q_i \leq p_i$ $\forall i \in h(p,\alpha)$ and $q_i \geq p_i$ $\forall i \notin h(p,\alpha)$,
then $h(p,\alpha) = h(q,\alpha)$.
  \item There exists $\delta>0$ such that $p \in [0,1]^m$, $\alpha \in [0,1]$ and
$\| p-p^\ast \| \vee |\alpha-\alpha^\ast| < \delta$ imply $h(p,\alpha)=h(p^\ast,\alpha^\ast)$.
\end{enumerate}
\end{condition}

Condition \ref{condition_h} ensures that lowering (increasing) the p-value
of any rejected (non-rejected) hypothesis does not affect the result of $h$.
Moreover, we require that there
exists a neighborhood of $p^\ast$ and $\alpha^\ast$
on which $h$ is constant.
In Section \ref{section_procedures} we will simplify Condition \ref{condition_h}
for so-called step-up and step-down procedures.

We will call a monotonic multiple testing procedure $h$ \textit{well-behaved}
for $p^\ast$ and $\alpha^\ast$ if it satisfies Condition \ref{condition_h}.
The multiple testing procedures we consider in this article,
such as the common procedures of
\cite{Bonferroni1936}, \cite{Sidak1967}, \cite{Holm1979}, \cite{Hochberg1988}
or the one of \cite{Benjamini1995CFD}
are well-behaved
for all but a null set of $p^\ast$ and $\alpha^\ast$ (with respect to the Lebesgue measure).
A (non-exhaustive) list of well-behaved procedures can be found in Section \ref{subsection_procedures_examples}.

A second condition is necessary to obtain convergence of the two
bounds $\underline{A}_n$ and $\overline{A}_n$
established in Lemma \ref{lemma_monotonicity}
to $h(p^\ast,\alpha^\ast)$ as $n \rightarrow \infty$.
Whereas on the event $\event_1$, all hypotheses in $\underline{A}_n$ ($\overline{A}_n^c$)
can already be rejected (non-rejected),
we additionally require that
the length of each interval
belonging to a yet undecided hypothesis in the set
$\overline{A}_n \setminus \underline{A}_n$
or to the threshold goes to zero:
$$\event_2 = \left\{ \max \{ |F_i(O_{1:n})|: i \in \overline{A}_n \setminus \underline{A}_n \cup \{m+1\} \}
\rightarrow 0 \text{ as } n \rightarrow \infty \right\}.$$
The following theorem improves upon Lemma \ref{lemma_monotonicity}
on the more restrictive event $\event=\event_1 \cap \event_2$:

\begin{theorem}
\label{theorem_convergence}
Let $h$ be a well-behaved multiple testing procedure for $p^\ast$ and $\alpha^\ast$.
On the event $\event$, both sequences
$(\underline{A}_n)_{n \in \N}$ and $(\overline{A}_n)_{n \in \N}$
converge to $h(p^\ast,\alpha^\ast)$,
i.e.\ there exists $n_0 \in \N$ such that
$\underline{A}_n = h(p^\ast,\alpha^\ast) = \overline{A}_n$
$\forall n \geq n_0$.
\end{theorem}

In the next section, we will use Lemma \ref{lemma_monotonicity} and
Theorem \ref{theorem_convergence} to establish
guarantees on the test result of existing algorithms.

Suppose Algorithm \ref{algorithm_generic} is used in connection with a
well-behaved multiple testing procedure controlling the familywise error rate (fwer).
Then at any stage, the fwer is also controlled for
all the rejections in $\underline{A}_n \subseteq h(p^\ast,\alpha^\ast)$.
This is easily proven using Boole's inequality.

A similar statement, however, is not true for
well-behaved multiple testing procedures controlling the false discovery rate (fdr).
Although the fdr is not generally controlled for subsets
$\underline{A}_n \subseteq h(p^\ast,\alpha^\ast)$
or supersets $\overline{A}_n \supseteq h(p^\ast,\alpha^\ast)$,
the following guarantees hold
if Algorithm \ref{algorithm_generic} is run with suitable stopping times.

\begin{lemma}
\label{lemma_fdr}
Let $h$ control the fdr at level $\alpha$, let $\underline{V}_n$ ($\overline{V}_n$)
be the set of rejected true null hypotheses in
$\underline{A}_n$ ($\overline{A}_n$) for $n \in \N$
and let $\eta \geq 1$, $\xi \geq 0$.
\begin{enumerate}
  \item $\E \left( |\underline{V}_s|/|\underline{A}_s| \right) \leq \eta \alpha$ for the stopping time
  $s=\min \{n \in \N: |\overline{A}_n|/|\underline{A}_n| \leq \eta \}$.
  \item $\E \left( |\overline{V}_t|/|\overline{A}_t| \right) \leq \alpha + \xi$ for
  $t=\min \{n \in \N: (|\overline{A}_n|-|\underline{A}_n|)/|\overline{A}_n| \leq \xi \}$.
\end{enumerate}
\end{lemma}

In Lemma \ref{lemma_fdr}, we define the fraction in the definition of the stopping time
$s$ (time $t$) to be zero if $|\underline{A}_n|$ ($|\overline{A}_n|$) is zero
as in this case, false rejection errors are impossible.
Lemma \ref{lemma_fdr} thus provides two different guarantees on the fdr,
a multiplicative one on the set of rejected hypotheses $\underline{A}_s$ and an additive
guarantee on the rejections in $\overline{A}_t$ with respect to the two stopping times $s$ and $t$.

\section{Improving existing algorithms}
\label{section_improving}
In this section, we introduce a class of established methods
which estimate $h(p^\ast,\alpha^\ast)$
and show how the framework
given by Algorithm \ref{algorithm_generic}
can be used to modify these methods in such a way as to provide a guarantee
on the correctness of their test results.
We will demonstrate our proposed modifications
by extending the improved na\"ive method presented in
Section \ref{subsection_improved_naive}
to the situation of an estimated testing threshold.

Consider an existing method to compute $h(p^\ast,\alpha^\ast)$.
The threshold $\alpha^\ast$ can either be constant or given by
a monotonic (increasing or decreasing) function $g: [0,1]^m \rightarrow \R$,
thus $\alpha^\ast=g(p^\ast)$.
In the latter case, $\alpha^\ast$ is a function of $p^\ast$ and thus unknown itself.

Methods working with bootstrap point estimates of $p^\ast$
\citep{BesagClifford1991,Wieringen2008,Sandve2011,JiangSalzman2012},
fitted distributions \citep{Knijnenburg2009}
or permutation based methods
\citep{WestfallYoung1993,WestfallTroendle2008,Meinshausen2006}
can be phrased in the following way:
Draw independent samples $X_{ij}\sim$ Bernoulli($p_i^\ast$), $j \in \N$, for each $i \in \{ 1,\ldots,m \}$.
Use a finite number $S_i$ of these samples $X_{i1},\ldots,X_{i,S_i}$
to compute a p-value estimate $\hat{p}_i$ of $p_i^\ast$,
where $S_i$ is a (random) index and $i \in \{ 1,\ldots,m \}$.
Estimate the testing threshold $\alpha^\ast$ using the plug-in estimate $\hat{\alpha}=g(\hat{p})$,
where $\hat{p} = (\hat{p}_1,\ldots,\hat{p}_m )$.
Return $h(\hat{p},\hat{\alpha})$ as the test result.

Based on Algorithm \ref{algorithm_generic}
we propose to modify any method of the above type by
\begin{enumerate}
  \item Maintaining a confidence sequence \citep{Lai1976}
with a coverage probability of $1-\epsilon/m$
for each p-value $p_i^\ast$, $i \in \{ 1,\ldots,m \}$,
and by using each sequence as $F_i(O_{1:n})$ in Algorithm \ref{algorithm_generic}.
The overall error probability $\epsilon$ is chosen by the user.
  \item Computing plug-in bounds $F_{m+1}(O_{1:n})$ for $\alpha^\ast$
using the monotonicity of $g$ and the above confidence sequences.
  \item Reporting hypotheses in $\underline{A}_n$ as rejected and
  in $\overline{A}_n^c$ as non-rejected.
  The remaining hypotheses are still undecided.
\end{enumerate}

As the confidence sequence of \cite{Lai1976} satisfies
$\PP(\exists n: p_i^\ast \notin F_i(O_{1:n})) < \beta$
for each $p_i^\ast$
(see Section \ref{subsection_improved_naive}),
the choice $\beta=\epsilon/m$ yields
\begin{align*}
\PP(\exists i,n:p_i^\ast \notin F_i(O_{1:n}))
&\leq \sum_{i=1}^m \PP(\exists n: p_i^\ast \notin F_i(O_{1:n}))
\leq \sum_{i=1}^m \epsilon/m = \epsilon,
\end{align*}
and hence $\PP(p_i^\ast \in F_i(O_{1:n}) ~\forall i \in \{ 1,\ldots,m \}, n \in \N) \geq 1-\epsilon$.
The event $\event_1$ thus occurs with probability at least $1-\epsilon$.

Consequently, any modified method of the above type
has the following advantage over its unimproved counterpart:

\begin{remark}
\label{remark_improved_method}
By Lemma \ref{lemma_monotonicity},
a modified method of the above type has the property that
all the hypotheses in the set $\underline{A}_n$ ($\overline{A}_n^c$)
which are rejected (non-rejected) in any iteration $n$
are indeed correctly rejected (non-rejected) with probability at least $1-\epsilon$.
\end{remark}

Remark \ref{remark_improved_method} applies to the improved na\"ive method
(Section \ref{subsection_improved_naive})
upon stopping in iteration $s$ as well as to the
methods presented in the following two examples.
First, we generalize
Section \ref{subsection_improved_naive}
to the situation where the testing threshold is unknown.

\begin{example}
\label{example_poundscheng}
\normalfont
Additionally to the setting of
Section \ref{subsection_improved_naive},
we assume that multiple testing is
carried out at the corrected testing threshold
$\alpha^\ast = t^\ast/ \hat{\pi}_0(p^\ast)$,
where $t^\ast$ is an uncorrected threshold
(typically $t^\ast=0.05$ or $t^\ast=0.1$)
and $\hat{\pi}_0(p)=\min \left( 1,\frac{2}{m} \sum_{i=1}^m p_i \right)$
is an estimator of the proportion of true null hypotheses \citep{PoundsCheng2006}.
Recent applications of this threshold include
\cite{HanDalal2012}, \cite{Lu2011}, \cite{Jupiter2010}, \cite{Cheng2009}.
As $\hat{\pi}_0(p^\ast)$ depends on the p-values,
the corrected threshold $\alpha^\ast$ is unknown in practice.
We thus need to compute a confidence interval for it.
The interval can be constructed using the monotonicity of $\hat{\pi}_0(p)$:
in iteration $n$,
$\underline{\pi}_n=\hat{\pi}_0( \min I_n^1,\ldots,\min I_n^m  )$
is a lower bound on $\hat{\pi}_0(p^\ast)$, likewise
$\overline{\pi}_n=\hat{\pi}_0( \max I_n^1,\ldots,\max I_n^m )$
is an upper bound.
This immediately translates to the interval
$F_{m+1}(O_{1:n})=[ t^\ast/\overline{\pi}_n, t^\ast/\underline{\pi}_n ]$
for $\alpha^\ast$.
\end{example}

We try to improve Example \ref{example_poundscheng}
by using a (hopefully) tighter confidence interval $F_{m+1}(O_{1:n})$ for $\alpha^\ast$
based on Hoeffding's inequality \citep{Hoeffding1963}.

\begin{example}
\label{example_hoeffding}
\normalfont
Suppose we have observed $s$ samples $X_1^i,\ldots,X_s^i$ per hypothesis $H_{0i}$,
where $X^i_j$ is the indicator
of an exceedance for the $j$th sample drawn for $H_{0i}$
(see Section \ref{subsection_improved_naive}).
Then,
$\PP \left( \left| \frac{1}{ms} \sum_{i=1}^m \sum_{j=1}^s X^i_j
- \frac{1}{m} \sum_{i=1}^m p_i^\ast \right| \geq u \right)
\leq 2 \exp \left( -2 ms u^2 \right)$
for all $u>0$ by Hoeffding's inequality.
Thus for a given $\eta \in [0,1]$,
$\frac{1}{ms} \sum_{i=1}^m \sum_{j=1}^s X^i_j \pm \sqrt{ -\log(\eta/2)/(2ms) }$
are boundaries of a $1-\eta$ confidence interval for
$\frac{1}{m} \sum_{i=1}^m p_i^\ast$.
Using the monotonicity of the mapping $x \mapsto t^\ast/\min(1,2x)$,
this immediately translates to a $1-\eta$ confidence interval for $\alpha^\ast$.
When using Hoeffding's interval in the improved na\"ive method,
we allocate an error of $\eta=\epsilon/(m+1)$ to the computation of Hoeffding's interval
as well as to the computation of each of the $m$ confidence sequences for the p-values.
As the improved na\"ive method is open-ended, we use a non-negative real sequence
$(\eta_n)_{n \in \N}$ satisfying $\sum_{n=1}^\infty \eta_n=\eta$
to distribute $\eta$ for Hoeffding's interval over all iterations of the algorithm,
thus computing it at level $\eta_n$ in each iteration $n$.
\end{example}

Both the plug-in interval (Example \ref{example_poundscheng})
and Hoeffding's confidence interval (Example \ref{example_hoeffding})
will be evaluated in Section \ref{section_use_practice}.

\section{Well-behaved step-up and step-down procedures}
\label{section_procedures}
Although the multiple testing procedure $h$ does not have to be of
a special form, many procedures used in practice such as the ones of
\cite{Bonferroni1936}, \cite{Sidak1967}, \cite{Holm1979}, \cite{Hochberg1988}
or the one of \cite{Benjamini1995CFD}
belong to a certain class of procedures,
called step-up and step-down procedures.
We will simplify Condition \ref{condition_h} for step-up and step-down procedures
in Section \ref{subsection_procedures_stepupstepdown}
and use the simplified condition
in Section \ref{subsection_procedures_examples}
to verify that many widely used procedures are well-behaved.
As shown in Section \ref{appendix_hommel} in the Appendix, the \cite{Hommel1988} procedure
is an example of a procedure which is not well-behaved.

\subsection{Condition \ref{condition_h} can be simplified for step-up and step-down procedures}
\label{subsection_procedures_stepupstepdown}
Suppose we are given an arbitrary step-up procedure $h_u$ or step-down procedure $h_d$
\citep{RomanoShaikh2006}
returning the set of rejected indices.
For our purposes, we phrase these two procedures
in terms of a threshold function
$\tau_\alpha: \{ 1,\ldots,m \} \rightarrow [0,1]$
which depends on a threshold $\alpha \in [0,1]$
and returns the critical value $\tau_\alpha(i)$
each $p_{(i)}$ is compared to:
\begin{align}
\label{label_stepup}
h_u(p,\alpha) &= \left\{ i \in \{ 1,\ldots,m \}: p_i \leq \max \{ p_{(j)}: p_{(j)} \leq \tau_\alpha(j) \} \right\},\\
\label{label_stepdown}
h_d(p,\alpha) &= \left\{ i \in \{ 1,\ldots,m \}: p_i < \min \{ p_{(j)}: p_{(j)} > \tau_\alpha(j) \} \right\},
\end{align}
where $\max \emptyset := 0$, $\min \emptyset := 1$,
and where the order statistic of $p_1,\ldots,p_m$
is denoted by $p_{(1)} \leq \ldots \leq p_{(m)}$.

We assume that the threshold function $\tau_\alpha$ satisfies the following condition.

\begin{condition}
\label{condition_threshold}
\begin{enumerate}
  \item $\tau_\alpha(i)$ is non-decreasing in $i$ for each fixed $\alpha$.
  \item $\tau_\alpha(i)$ is continuous in $\alpha$ and non-decreasing in $\alpha$ for each fixed $i$.
\end{enumerate}
\end{condition}

By the following lemma, a step-up or step-down procedure
is well-behaved if the threshold function $\tau_\alpha$
defining it satisfies Condition \ref{condition_threshold}.

\begin{lemma}
\label{lemma_stepupstepdown}
If $\tau_\alpha$ satisfies Condition \ref{condition_threshold}
then the corresponding $h_u$ and $h_d$ are monotonic
and satisfy the first part of Condition \ref{condition_h}.
If moreover $\tau_{\alpha^\ast}(i) \neq p^\ast_{(i)}$ for all $i \in \{ 1,\ldots,m \}$,
$h_u$ and $h_d$ also satisfy the second part of Condition \ref{condition_h}
for $p^\ast$ and $\alpha^\ast$.
\end{lemma}

We investigate in which cases the condition
$\tau_{\alpha^\ast}(i) \neq p^\ast_{(i)}$ for all $i \in \{ 1,\ldots,m \}$
in Lemma \ref{lemma_stepupstepdown}
is satisfied if $p^\ast$ are random.

If $p^\ast$ come from a discrete distribution,
the p-values satisfying
$\tau_{\alpha^\ast}(i) \neq p^\ast_{(i)}$
do not necessarily form a null set.
For a fixed $\alpha^\ast$,
however, the p-values not satisfying the conditions of
Lemma \ref{lemma_stepupstepdown} form a null set
if $p^\ast$ are random with a distribution
that is absolutely continuous with respect to the Lebesgue measure.

We now consider the case of a threshold $\alpha^\ast$ given by a deterministic function of the p-values $p^\ast$.
We show that
the p-values $p^\ast$ not satisfying the condition
$\tau_{\alpha^\ast}(i) \neq p^\ast_{(i)}$ for all $i \in \{ 1,\ldots,m \}$
in Lemma \ref{lemma_stepupstepdown}
form a null set
if $p^\ast$ come from an absolutely continuous distribution with respect to the Lebesgue measure,
and if the \cite{Benjamini1995CFD} or \cite{Bonferroni1936}
procedure applied to the $p^\ast$ dependent threshold of \cite{PoundsCheng2006} is used to test the hypotheses.

The \cite{Benjamini1995CFD} procedure
is characterized by
the threshold function $\tau_\alpha(i)=i\alpha/m$
(see Section \ref{subsection_procedures_examples}).
The threshold of \cite{PoundsCheng2006} is given by
$\alpha^\ast(p^\ast) = t^\ast/\min \left( 1, \frac{2}{m} \sum_{r=1}^m p_r^\ast \right)$
(see Examples \ref{example_poundscheng} and \ref{example_hoeffding}).

Fix $i \in \{1,\ldots,m\}$.
We use the fact that $\sum_{r=1}^m p_r^\ast = \sum_{r=1}^m p_{(r)}^\ast$ and
that either
$\tau_{\alpha^\ast(p^\ast)}(i) = it^\ast/m$
or
$\tau_{\alpha^\ast(p^\ast)}(i) = \frac{i t^\ast}{2}(\sum_{r=1}^m p_r^\ast)^{-1}$.
Conditional on $\{ p^\ast_{(r)}: r \neq i \}$,
$$\tau_{\alpha^\ast(p^\ast)}(i)=p_{(i)}^\ast
\Rightarrow
\frac{it^\ast}{m}=p_{(i)}^\ast \vee \frac{it^\ast}{2}=p_{(i)}^\ast(p_{(i)}^\ast+s_i)
\Leftrightarrow
p_{(i)}^\ast \in \left\{ \frac{it^\ast}{m}, \rho_1, \rho_2 \right\},$$
where $s_i = \sum_{r\neq i} p_{(r)}^\ast$ and $\rho_1$, $\rho_2$ are the two solutions
of $it^\ast/2=\rho(\rho+s_i)$.
Two distinct solutions always exist given $t^\ast>0$.

Using that
$p^\ast_{(1)} \leq \cdots \leq p^\ast_{(m)}$
are also random with a distribution that is absolutely continuous
with respect to the Lebesgue measure,
given $p^\ast$ come from an absolutely continuous distribution, implies
$\PP(p_{(i)}^\ast \in \{ it^\ast/m, \rho_1, \rho_2 \} ~|~ p^\ast_{(r)}: r \neq i )=0$.

The previous result immediately extends to
$\PP(\exists i: \tau_{\alpha^\ast(p^\ast)}(i)=p_{(i)}^\ast)=0$,
hence the p-values $p_{(i)}^\ast$ which coincide with their critical value $\tau_{\alpha^\ast(p^\ast)}(i)$
form a null set.

As the threshold function of the \cite{Bonferroni1936} correction can be
recovered from the one of the \cite{Benjamini1995CFD} procedure
by removing the dependence of $\tau_\alpha(i)=i\alpha/m$
on $i$ (see Section \ref{subsection_procedures_examples}),
the above result also holds true for the \cite{Bonferroni1936} correction.

A similar argumentation can be used to extend the above result
to other common estimators of $\alpha^\ast$ and threshold functions $\tau_\alpha$.

\subsection{Examples of well-behaved step-up and step-down procedures}
\label{subsection_procedures_examples}
This section shows that a variety of commonly used step-up and
step-down procedures are monotonic and satisfy Condition \ref{condition_threshold}.

The following multiple testing procedures are determined by $\tau_\alpha(i)$,
where $i \in \{ 1,\ldots,m \}$,
and control the fwer or the fdr at a threshold $\alpha$.
We denote the hypothesis corresponding
to the ordered p-value $p_{(i)}$ by $H_{0(i)}$, $i \in \{ 1,\ldots,m \}$.

In most cases, Condition \ref{condition_threshold} can be checked
by considering the derivatives of $\tau_\alpha(i)$
with respect to $\alpha$ and $i$, thus regarding $i$ as a continuous parameter.
Unless stated otherwise, all the threshold functions listed below are clearly
non-decreasing in both $i$ and $\alpha$ as well as continuous in $\alpha$
and thus satisfy Condition \ref{condition_threshold}.

The \cite{Bonferroni1936} correction can be derived from either
a step-up or a step-down procedure
using the constant threshold function $\tau_\alpha(i)=\alpha/m$.

The following step-up procedures are well-behaved:
\begin{enumerate}
  \item The \cite{Simes1986} procedure rejects $\cap_{i \in \{ 1,\ldots,m \} } H_{0i}$
if there exists $k \in \{ 1,\ldots,m \}$ such that $p_{(k)} \leq k \alpha/m$.
It can be used in our framework with the help of the following modification:
Once $h_u(p,\alpha)$ for a step-up procedure with threshold function $\tau_\alpha(i)=i\alpha/m$
is correctly determined, the \cite{Simes1986} procedure rejects
$\cap_{i \in \{ 1,\ldots,m \} } H_{0i}$
if and only if $|h_u(p,\alpha)|>0$.

  \item The \cite{Hochberg1988} procedure uses $\tau_\alpha(i) = \alpha/(m+1-i)$.

  \item The \cite{Rom1990} procedure increases the power
of the \cite{Hochberg1988} procedure by replacing its critical
values $\tau_\alpha(i) = \alpha/(m+1-i)$ by ``sharper'' values
$\tau_\alpha(i) = c_i$. The $c_i$ are computed recursively as given
in \cite{Rom1990} and satisfy $c_i \nearrow$ for
a fixed $\alpha$. Moreover, the $c_i$ are non-decreasing in $\alpha$.

  \item The choice $\tau_\alpha(i)=i\alpha/m$ yields the \cite{Benjamini1995CFD} procedure.

  \item The \cite{BenjaminiYekutieli2001} procedure controls the
fdr under arbitrary dependence by applying the
\cite{Benjamini1995CFD} procedure at the corrected constant threshold
$\alpha/\left( \sum_{i=1}^m i^{-1} \right)$.
\end{enumerate}

Similarly, the following step-down procedures satisfy Condition \ref{condition_threshold}:
\begin{enumerate}
  \item The \cite{Sidak1967} correction uses $\tau_\alpha(i) = 1-(1-\alpha)^{1/(m+1-i)}$.

  \item The choice $\tau_\alpha(i) = \alpha/(m+1-i)$
yields the \cite{Holm1979} procedure.

  \item The \cite{Shaffer1986} procedure modifies the
\cite{Holm1979} procedure in order to obtain an increase in power.
For the tests under consideration,
let $0 \leq a_1 < a_2 < \cdots < a_r \leq n$ be
all possible numbers of true null hypotheses.
Assuming that $H_{0(1)}, \ldots, H_{0(i-1)}$ are false,
let $t_i = \max \{ a_j: a_j \leq n-i+1 \}$ be the maximum possible number of
true null hypotheses.
The \cite{Shaffer1986} procedure determines the minimal index
$k$ such that $p_{(k)} > \alpha/t_k$ and then rejects
$H_{0(1)}, \ldots, H_{0(k-1)}$.
It can be obtained from a step-down procedure using
$\tau_\alpha(i) = \alpha/t_i$,
which is clearly continuous and non-decreasing in $\alpha$ for a fixed $i$.
As $a_i \nearrow$ and thus $t_i \searrow$,
$\tau_\alpha(i)$ is also non-decreasing in $i$ for a fixed $\alpha$.
\end{enumerate}

For a given $\alpha^\ast$,
by Lemma \ref{lemma_stepupstepdown}, all the procedures listed above
are well-behaved for all but a null set of p-values $p^\ast$.

\section{Using the framework in practice}
\label{section_use_practice}
The improved na\"ive method
(Section \ref{subsection_improved_naive})
is capable of computing test results which consist,
up to a pre-specified error probability $\epsilon$,
of sets of correctly rejected and correctly non-rejected hypotheses
as well as of a set of undecided hypotheses.
The following contains an example of such a classification.

\cite{Sandve2011} use their method \texttt{MCFDR} to classify a dataset of
gene modifications (so-called H3K4me2-modifications)
of \cite{Pekowska2010}.
This dataset consists of gene regions and
gene modifications within each region, characterized by their midpoint.
The beginning and the end of each
region on the genome are normed to $0$ and $1$, respectively.
The authors test if the gene modifications appear
more often in a certain part of the gene region.

To be precise,
\cite{Sandve2011} observe $k$ random points $Y_1,\ldots,Y_k$ in $[0,1]$
(these are the midpoints of the gene modifications)
and test the null hypothesis
$H_0: \E \left( \frac{1}{k} \sum_{i=1}^k Y_i \right) \geq 0.5$
against the alternative
$H_1: \E \left( \frac{1}{k} \sum_{i=1}^k Y_i \right) < 0.5$
using the test statistic $T=\frac{1}{k} \sum_{i=1}^k Y_i$.
Each null hypothesis is tested by permuting
the midpoints in each region while preserving
their inter-point distances.

\cite{Sandve2011} first filter the dataset for genes with at least $10$ modifications
per gene region.
Each such region becomes one hypothesis,
leading to $m=3465$ hypotheses (gene regions) under consideration.
They evaluate the data using the procedure of \cite{Benjamini1995CFD}
with a corrected testing threshold at level
$0.1/\hat{\pi}_0(\hat{p})$, where $\hat{\pi}_0$ is the
estimator of \cite{PoundsCheng2006}
introduced in Example \ref{example_poundscheng}
and $\hat{p}$ is an estimate of $p^\ast$ returned by \texttt{MCFDR}.
\cite{Sandve2011} report $2747$ significant hypotheses.

Nevertheless, the authors do not provide any
guarantee on the correctness of their findings.
Recomputing the results of \cite{Sandve2011} indeed shows
considerable variability.
To demonstrate this,
we re-classify the H3K4me2 dataset using the \texttt{MCFDR} algorithm
of \cite{Sandve2011} a total number of $r=1000$ times.
Let $p_i^s$ ($p_i^n$)
be the empirical probability that hypothesis $H_{0i}$
is significant (non-significant) in these $r$ repetitions.

We are interested in measuring the randomness in the output of an algorithm
and use $p_i^r = \min (p_i^s,p_i^n)$
as probability of $H_{0i}$ being randomly classified.
We call all hypotheses having
$p_i^r>0.01$
``randomly classified'' and denote their total number by \textit{rc}.
The choice $0.01$ is arbitrary.
It depends on how much uncertainty
a user is willing to tolerate for a single decision on a hypothesis to be ``reasonably firm''.
For \texttt{MCFDR} we observe that $353$ hypotheses
remain randomly classified on average.

\begin{table}[t]
\begin{center}
\caption{Repeated application of the improved and the unimproved
na\"ive method to the same data.
\label{table_comparison}}
\begin{tabular}{l|l|llll|llll}
& na\"ive & \multicolumn{8}{c}{improved na\"ive method}\\
& method & \multicolumn{4}{c|}{with plug-in interval (Ex.\ref{example_poundscheng})}
& \multicolumn{4}{c}{with Hoeffding's interval (Ex.\ref{example_hoeffding})}\\
s & \textit{rc} & rejected & non-rej. & undec. & \textit{rc} & rejected & non-rej. & undec. & \textit{rc}\\
\hline
$10^2$		&349		&0	&161.8	&3303.2	&0	&0	&372.0	&3093.0	&0\\
$10^3$		&107		&2386.0	&487.5	&591.5	&0	&2568.5	&576.0	&320.5	&0\\
$10^4$		&33		&2649.0	&624.6	&191.4	&0	&2697.3	&661.7	&106.0	&0
\end{tabular}
\end{center}
$s$: number of samples drawn per hypotheses;
\textit{rc}: number of randomly classified hypotheses;
rejected, non-rejected and undecided are average numbers
based on $1000$ repetitions.
\end{table}

We first use the (unimproved) na\"ive method
(as defined at the beginning of
Section \ref{subsection_improved_naive})
with $s \in \{ 10^2,10^3,10^4 \}$ samples per hypothesis
to classify the same dataset.
Table \ref{table_comparison} shows the number of randomly classified hypotheses
\textit{rc} observed for the na\"ive method
as a function of $s$ (second column).
For $s=10^2$, the total effort is comparable to the one of \texttt{MCFDR}
and both methods yield equally high numbers of random decisions
($rc \approx 350$).
For high precision ($s=10^4$), up to
$33$
hypotheses remain inconsistently classified.

We then apply the improved na\"ive method
(Section \ref{subsection_improved_naive})
to the same dataset
using an overall error probability of $\epsilon=0.01$.
The improved method is stopped after having drawn $s$ samples per hypothesis.
Table \ref{table_comparison} shows rejected, non-rejected,
undecided (see Remark \ref{remark_improved_method})
and randomly classified hypotheses.
We evaluate both the plug-in interval for $\alpha^\ast$
introduced in Example \ref{example_poundscheng} (columns three to six)
as well as Hoeffding's confidence interval
derived in Example \ref{example_hoeffding} (columns seven to ten).
For Hoeffding's interval, we use
$\eta_n=\nu_n-\nu_{n-1}$ with $\nu_n=\frac{n}{n+s} \frac{\epsilon}{m+1}$, $n \in \N$.

Using a confidence interval for $\alpha^\ast$ based on Hoeffding's inequality
(as opposed to the plug-in interval)
yields considerably more decisions (rejections and non-rejections) and thus
less undecided hypotheses for all ranges of precision.

Although for low numbers of samples many hypotheses remain undecided,
the test results of the improved na\"ive method are consistent
in the sense that no hypothesis is randomly classified.
The improved na\"ive method therefore provides reliable test results
and ensures repeatability.
For a high precision ($s=10^4$), the
improved na\"ive method with Hoeffding's interval for $\alpha^\ast$ yields around
$2700$ rejections and $660$ non-rejections.
The remaining $106$ hypotheses are still undecided,
meaning that within this limited computational effort, no
statement about these hypotheses (gene regions) should be made.
The probability of the above results being correct is at least $0.99$.

We interpret the set of undecided hypotheses as the set of hypotheses for which
a clear decision exists,
even though this decision cannot yet be obtained within the limited computational effort used in a real testing scenario.
Consequently, by using more Monte Carlo samples, the decision of any hypothesis will eventually be revealed.
Alternatively, one can view the set of undecided hypotheses as a set of hypotheses
whose decision is essentially arbitrary.

Finally, the framework is not
limited to a strict familywise error control on all its Monte Carlo decisions.
It would be possible to relax the guarantee and to control a less conservative criterion instead,
for instance the false discovery rate.

\section{Discussion}
\label{section_discussion}
The present article considers p-value based multiple
testing under the assumption that
the p-value of each hypothesis is unknown
and can only be approximated using Monte Carlo simulations.
Although widely occurring in experimental studies,
common methods for this scenario
do not give any guarantee on how their test results
relate to the one obtained if all p-values had been known.

The article introduces a framework for Monte Carlo based multiple testing,
both in terms of a general multiple testing procedure and a generic algorithm.
Conditions on both the multiple testing procedure and the algorithm
guarantee that the rejections and non-rejections
returned by our generic algorithm
are identical to the ones obtained with the p-values.
A simplified condition for step-up and step-down multiple testing procedures is derived.

We demonstrate how to use our framework
to modify established methods in such a way
as to yield theoretical guarantees on their test results.
As demonstrated on a class of commonly used methods,
these modifications can easily be implemented in practice and
come at virtually no additional computational cost.

Improved established methods,
such as the improved na\"ive method evaluated in this article on a real data study,
allow one to report multiple testing results as three sets:
rejected, non-rejected and undecided hypotheses, together
with an error bound on their correctness.
We recommend any multiple testing result to be reported in this fashion.

\appendix
\section{Proofs}
\label{appendix_proofs}
For simplicity of notation we sometimes drop the dependence of the
multiple testing procedure $h(p,\alpha)$ on the threshold $\alpha$.

\subsection{Proofs of Section \ref{subsection_framework_convergenceresults}}
\label{appendix_proofs_subsection_framework_convergenceresults}
\begin{proof}[Proof of Lemma \ref{lemma_monotonicity}]
1. By construction, Algorithm \ref{algorithm_generic}
computes nested intervals, thus
$\overline{p}_n = (\max I_n^i)_{i \in \{1,\ldots,m\} } \searrow$ and
$\overline{\alpha}_n = \max I_n^{m+1} \searrow$
as well as
$\underline{p}_n = (\min I_n^i)_{i \in \{1,\ldots,m\} } \nearrow$ and
$\underline{\alpha}_n = \min I_n^{m+1} \nearrow$.
Hence,
\begin{align*}
\underline{A}_n = h(\overline{p}_n,\underline{\alpha}_n) \subseteq h(\overline{p}_{n+1},\underline{\alpha}_n)
\subseteq h(\overline{p}_{n+1},\underline{\alpha}_{n+1}) = \underline{A}_{n+1},\\
\overline{A}_n = h(\underline{p}_n,\overline{\alpha}_n) \supseteq h(\underline{p}_{n+1},\overline{\alpha}_n)
\supseteq h(\underline{p}_{n+1},\overline{\alpha}_{n+1}) = \overline{A}_{n+1},
\end{align*}
where the first (second) subset relation follows from the monotonicity
of $h$ (Condition \ref{condition_h}) in the first (second) argument.

2. On the event $\event_1$,
$p_i^\ast \in I_n^i$ and $\alpha^\ast \in I_n^{m+1}$
for all $i$ and $n$, thus $\overline{p}_n \geq p^\ast_n \geq \underline{p}_n$
and $\underline{\alpha}_n \leq \alpha^\ast \leq \overline{\alpha}_n$.
By monotonicity of $h$ (Condition \ref{condition_h}),
$\underline{A}_n = h(\overline{p}_n,\underline{\alpha}_n) \subseteq
h(p^\ast,\alpha^\ast) \subseteq h(\underline{p}_n,\overline{\alpha}_n) = \overline{A}_n$
$\forall n \in \N$.
\end{proof}

\begin{proof}[Proof of Theorem \ref{theorem_convergence}]
Let $\overline{\alpha}_n = \max I_n^{m+1}$,
$\underline{\alpha}_n = \min I_n^{m+1}$
as well as
$B_n = \overline{A}_n \setminus \underline{A}_n$.
Suppose $\exists i \in \limsup_{n \rightarrow \infty} B_n$.
On the event $\event_2$,
$|I_n^i| \rightarrow 0$ as $n \rightarrow \infty$
for $i \in \limsup_{n \in \N} B_n$
as well as $|I_n^{m+1}| \rightarrow 0$ as $n \rightarrow \infty$.
Let $\delta$ be as given in Condition \ref{condition_h}.
As $B_n \subseteq \{1,\ldots,m \}$ is finite $\forall n \in \N$, there exists
$n_0 \in \N$ such that $|I_n^i|^2 < \delta^2/m$
and $|\underline{\alpha}_n - \overline{\alpha}_n| < \delta$ for
$n \geq n_0$ and all $i \in \lim \sup_{n \rightarrow \infty} B_n$.

We show that for all $n \geq n_0$,
\begin{align*}
\overline{A}_n =
h( (\min I_n^i)_{i \in \{ 1,\ldots,m \}}, \overline{\alpha}_n ) = h(p^\ast,\alpha^\ast) =
h( (\max I_n^i)_{i \in \{ 1,\ldots,m \}}, \underline{\alpha}_n )
=\underline{A}_n.
\end{align*}
To do this, we show
\begin{align*}
h(p^{(1)},\overline{\alpha}_n) &= h(p^{(2)},\overline{\alpha}_n) = h(p^{(3)},\overline{\alpha}_n),\\
h(p^{(3)},\overline{\alpha}_n) &= h(p^{(4)},\alpha^\ast) = h(p^{(5)},\underline{\alpha}_n),\\
h(p^{(5)},\underline{\alpha}_n) &= h(p^{(6)},\underline{\alpha}_n) = h(p^{(7)},\underline{\alpha}_n),
\end{align*}
where
\begin{center}
\begin{tabular}{lll}
$p^{(1)} := (\min I_n^i)_{i \in \{ 1,\ldots,m \}},$&
$p^{(4)} := p^\ast,$\\
$p^{(2)} := \begin{cases} \begin{matrix} \min I_n^i & i \in \overline{A}_n,\\ p_i^\ast & i \notin \overline{A}_n, \end{matrix} \end{cases}$&
$p^{(5)} := \begin{cases} \begin{matrix} \max I_n^i & i \in B_n,\\ p_i^\ast & i \notin B_n, \end{matrix} \end{cases}$\\
$p^{(3)} := \begin{cases} \begin{matrix} \min I_n^i & i \in B_n,\\ p_i^\ast & i \notin B_n, \end{matrix} \end{cases}$&
$p^{(6)} := \begin{cases} \begin{matrix} \max I_n^i & i \in \overline{A}_n,\\ p_i^\ast & i \notin \overline{A}_n, \end{matrix} \end{cases}$\\&
\end{tabular}
\end{center}
and
$p^{(7)} := (\max I_n^i)_{i \in \{ 1,\ldots,m \}}$.
The following holds true on the event $\event_1$.

(1) By definition, $\overline{A}_n = h(p^{(1)},\overline{\alpha}_n)$.
As $p^{(2)}_j = p_j^\ast \geq \min I_n^j = p^{(1)}_j$ $\forall j \notin \overline{A}_n$
and $p^{(2)}_j = p^{(1)}_j$ $\forall j \in \overline{A}_n$,
the first part of Condition \ref{condition_h} yields
$\overline{A}_n = h(p^{(1)},\overline{\alpha}_n) = h(p^{(2)},\overline{\alpha}_n)$
for a fixed $\overline{\alpha}_n$.

(2) As $(\max I_n^i)_{i \in \{ 1,\ldots,m \}} \geq p^{(3)}$ and as
$h$ is monotonic by Condition \ref{condition_h},
$\underline{A}_n \subseteq h(p^{(3)},\underline{\alpha}_n)
\subseteq h(p^{(3)},\overline{\alpha}_n)$.
As $p^{(2)}_j = \min I_n^j \leq p_j^\ast = p^{(3)}_j$ $\forall j \in \underline{A}_n$
and $p^{(2)}_j = p^{(3)}_j$ $\forall j \notin \underline{A}_n$,
the first part of Condition \ref{condition_h} yields
$h(p^{(2)},\overline{\alpha}_n) = h(p^{(3)},\overline{\alpha}_n)$.

(3) On the event $\event_1$,
$|\underline{\alpha}_n - \overline{\alpha}_n| < \delta$
implies $|\alpha^\ast - \overline{\alpha}_n| < \delta$
and $|I_n^i|^2 < \delta^2/m$ implies
$\| p^{(3)} - p^\ast \| < \delta$.
The second part of Condition \ref{condition_h} thus yields
$h(p^{(3)},\overline{\alpha}_n) = h(p^{(4)},\alpha^\ast) = h(p^\ast,\alpha^\ast)$
$\forall n \geq n_0$.

Arguing similarly to (1), (2), (3) we can show
$h(p^{(4)},\alpha^\ast) = h(p^{(5)},\underline{\alpha}_n)$
as well as
$h(p^{(5)},\underline{\alpha}_n) = h(p^{(6)},\underline{\alpha}_n)$
and
$h(p^{(6)},\underline{\alpha}_n) = h(p^{(7)},\underline{\alpha}_n) = \underline{A}_n$.
\end{proof}

\begin{proof}[Proof of Lemma \ref{lemma_fdr}]
Let $R$ be the set of rejected hypotheses and $V$ be the set of rejected true null hypotheses.
As $\underline{A}_n \subseteq R \subseteq \overline{A}_n$ for all $n \in \N$,
$|\underline{A}_n| \leq |R| \leq |\overline{A}_n|$.
Moreover, $\underline{V}_n \subseteq \underline{A}_n \subseteq R$ implies $|\underline{V}_n| \leq |V|$,
and as the difference in numbers of rejected true null hypotheses
in $V$ and $\overline{V}_n$ cannot differ by more than the number of undecided hypotheses
$|\overline{A}_n|-|\underline{A}_n|$ for any $n \in \N$,
$|V| \leq |\overline{V}_n| \leq |V| + (|\overline{A}_n|-|\underline{A}_n|)$.
\begin{enumerate}
  \item Using the above,
  $$\frac{|\underline{V}_n|}{|\underline{A}_n|} =
  \frac{|V|}{|R|} + \frac{|\underline{V}_n||R|-|\underline{A}_n||V|}{|\underline{A}_n||R|} \leq
  \frac{|V|}{|R|} + \frac{|V|}{|R|} \frac{|R|-|\underline{A}_n|}{|\underline{A}_n|} \leq
  \frac{|V|}{|R|} \frac{|R|}{|\underline{A}_n|} \leq
  \frac{|V|}{|R|} \frac{|\overline{A}_n|}{|\underline{A}_n|}$$
  for all $n \in \N$, thus
  $\E( |\underline{V}_s|/|\underline{A}_s| ) \leq \eta \E(|V|/|R|) = \eta \alpha$.
  \item Similarly,
  $$\frac{|\overline{V}_n|}{|\overline{A}_n|} \leq
  \frac{|V|+(|\overline{A}_n|-|\underline{A}_n|)}{|\overline{A}_n|} \leq
  \frac{|V|}{|R|}+\frac{|\overline{A}_n|-|\underline{A}_n|}{|\overline{A}_n|}$$
  for all $n \in \N$, thus
  $\E( |\overline{V}_t|/|\overline{A}_t| ) \leq \E(|V|/|R|) + \xi = \alpha + \xi$.
\end{enumerate}
\end{proof}

\subsection{Proofs of Section \ref{subsection_procedures_stepupstepdown}}
\label{appendix_proofs_subsection_framework_stepupstepdown}
The following two Lemmas will be needed for the proof of Lemma \ref{lemma_stepupstepdown}.
First, Lemma \ref{lemma_invariance} proves
three properties of step-up and step-down procedures
which are slightly stronger than the requirements
stated in Condition \ref{condition_h}.
For a vector $p=(p_1,\ldots,p_m)$,
we denote the rank of $p_i$ in the sorted sequence
$p_{(1)} \leq \ldots \leq p_{(m)}$ by $r_p(i)$.

\begin{lemma}
\label{lemma_invariance}
Let $p,q \in [0,1]^m$.
Let $h_u$ ($h_d$) be a step-up (step-down)
procedure defined through a threshold function
$\tau_\alpha$ satisfying Condition \ref{condition_threshold}.
\begin{enumerate}
  \item $h_u$ is monotonic.
  \label{lemma_item_monotonicity}
  \item If $q_i \leq \tau_\alpha(|h_u(p)|)$ $\forall i \in h_u(p)$
and $q_i = p_i$ $\forall i \notin h_u(p)$, then $h_u(p)=h_u(q)$.
  \label{lemma_item_invariance_rejection_area}
  \item If $q_i=p_i$ $\forall i \in h_u(p)$ and
$q_i > \tau_\alpha(r_p(i))$ $\forall i \notin h_u(p)$,
then $h_u(p) = h_u(q)$.\label{lemma_item_invariance_nonrejection_area}
  \item $h_d$ is monotonic.
  \label{lemma_item_monotonicity_d}
  \item If $q_i \leq \tau_\alpha(r_p(i))$ $\forall i \in h_d(p)$ and
$q_i=p_i$ $\forall i \notin h_d(p)$, then $h_d(p) = h_d(q)$.\label{lemma_item_invariance_rejection_area_d}
  \item If $q_i = p_i$ $\forall i \in h_d(p)$
and $q_i > \tau_\alpha(|h_d(p)|+1)$ $\forall i \notin h_d(p)$, then $h_d(p)=h_d(q)$.
  \label{lemma_item_invariance_nonrejection_area_d}
\end{enumerate}
\end{lemma}

\begin{proof}
As $h_u$ and $h_d$ are invariant to permutations,
we may assume $p_1 \leq \cdots \leq p_m$.

\ref{lemma_item_monotonicity}.
Let $p \in [0,1]^m$ and $i \in \{ 1,\ldots,m \}$.
It suffices to show that $h_u(p) \supseteq h_u(q)$ for any
$q \in [0,1]^m$ given by
$q_j=p_j$ $\forall j \neq i$ and $q_i > p_i$.

Let $k := |h_u(p)|$ be the largest rejected index.
We need to show that $j \notin h_u(q)$ $\forall j \geq k+1$.
Let $\alpha$ be fixed.

Case 1: $r_{q}(i) \leq k$. This implies $r_{q}(j)=j$ $\forall j \geq k+1$ and hence
$q_j = p_j > \tau_\alpha(j) = \tau_\alpha(r_{q}(j))$.
Therefore, $j \notin h_u(q)$ $\forall j \geq k+1$.

Case 2: $r_{q}(i) \geq k+1$. Let $j \geq k+1$, $j \neq i$.
Then the rank of the $j$th p-value can only drop by one when $p_i$
is replaced by $q_i$, i.e.\ $r_q(j) \in \{ j-1,j \}$.
Thus $q_j = p_j > \tau_\alpha(j) \geq \tau_\alpha(r_q(j))$ by Condition \ref{condition_threshold}
(using that $\tau_\alpha(i)$ is non-decreasing in $i$).
Furthermore, as $r_{q}(i) \geq k+1$, $q_i$
takes the position of the former $p_{r_{q}(i)}$ in the ordered sequence of
values from $q$, i.e.\ $q_i \geq p_{r_{q}(i)}$.
Hence, $r_{q}(i) \notin h_u(p)$
because of $r_{q}(i) \geq k+1$ and thus
$q_i \geq p_{r_{q}(i)} > \tau_\alpha(r_{q}(i))$.
Therefore, $\{ k+1, \ldots, m \} \cup \{ i \} \notin h_u(q)$.
This proves the monotonicity in the first argument of $h_u$.

The monotonicity in the second argument of $h$ is immediate as
$p_i \leq \max \{ p_{(j)}: p_{(j)} \leq \tau_\alpha(j) \}$
for all $i \in h_u(p,\alpha)$.
On Condition \ref{condition_threshold},
and using that $\tau_\alpha$ is non-decreasing in $\alpha$,
$\alpha \leq \alpha'$ implies $\tau_\alpha(j) \leq \tau_{\alpha'}(j)$ $\forall j$,
hence $i \in h_u(p,\alpha')$.
This proves \ref{lemma_item_monotonicity}.

\ref{lemma_item_invariance_rejection_area}.
All $i \notin h_u(p)$ satisfy $p_i > \tau_\alpha(r_p(i)) > \tau_\alpha(|h_u(p)|)$
whereas by assumption,
$q_i \leq \tau_\alpha(|h_u(p)|)$ $\forall i \in h_u(p)$.
Hence, using $q_i = p_i$ $\forall i \notin h_u(p)$, it follows that
$r_q(i)=r_p(i)$ $\forall i \notin h_u(p)$. Thus,
$q_i = p_i > \tau_\alpha(r_p(i)) = \tau_\alpha(r_q(i))$
for all $i \notin h_u(p)$. Hence $h_u(p)^c \subseteq h_u(q)^c$.

Conversely, define $\tilde{q} := \max \{ q_i: i \in h_u(p) \}$. As
$\tilde{q} \leq \tau_\alpha(|h_u(p)|) < q_i$
for all $i \notin h_u(p)$ and as there are precisely $|h_u(p)|$ values
$q_i \leq \tilde{q}$,
the rank of $\tilde{q}$ in $q$ is precisely $|h_u(p)|$.
As $q_i \leq \tilde{q} \leq \tau_\alpha(|h_u(p)|)$
$\forall i \in h_u(p)$,
all $\{ q_i \}_{i \in h_u(p)}$ are rejected, so $h_u(p) \subseteq h_u(q)$.
This proves \ref{lemma_item_invariance_rejection_area}.

\ref{lemma_item_invariance_nonrejection_area}.
As $q_i = p_i$ for all $i \in h_u(p)$, have $h_u(p) \subseteq h_u(q)$.

Let $i \notin h_u(p)$. If $r_q(i) \leq r_p(i)$, then
$q_i > \tau_\alpha(r_p(i)) \geq \tau_\alpha(r_q(i))$ by Condition \ref{condition_threshold}.
If $r_q(i) > r_p(i)$, $q_i$ replaces a
$q_j > \tau_\alpha(r_p(j))$ at rank $r_p(j)$
in the sorted sequence of $q$, hence $r_q(i)=r_p(j)$ and
$q_i \geq q_j > \tau_\alpha(r_p(j)) = \tau_\alpha(r_q(i))$.
Thus $q_i > \tau_\alpha(r_q(i))$ $\forall i \notin h_u(p)$,
which implies $h_u(p)^c \subseteq h_u(q)^c$.
This proves \ref{lemma_item_invariance_nonrejection_area}.

In a similar fashion,
\ref{lemma_item_monotonicity_d}.,
\ref{lemma_item_invariance_rejection_area_d}.
and \ref{lemma_item_invariance_nonrejection_area_d}.\
can be proven for step-down procedures $h_d$.
\end{proof}

For step-up procedures $h_u$, part \ref{lemma_item_invariance_rejection_area}.\ of
Lemma \ref{lemma_invariance}
shows that p-values of rejected hypotheses can be
increased up to $\tau_\alpha(|h_u(p)|)$,
the threshold evaluated at the last rejected hypothesis,
without affecting the result of $h_u$.
Part \ref{lemma_item_invariance_nonrejection_area}.\ of
Lemma \ref{lemma_invariance}
shows that $h_u$ is invariant
if p-values in the non-rejection area
are replaced by arbitrary values above the threshold.

Similarly, step-down procedures $h_d$ are invariant if p-values of rejected hypotheses are replaced
by arbitrary values below the threshold
(part \ref{lemma_item_invariance_rejection_area_d}.)
or p-values of non-rejected hypotheses are replaced by arbitrary
values above $\tau_\alpha(|h_d(p)|+1)$, the threshold evaluated
at the first non-rejected hypothesis
(part \ref{lemma_item_invariance_nonrejection_area_d}.).

The following Lemma \ref{lemma_invariance_delta} will also be needed
for the proof of Lemma \ref{lemma_stepupstepdown}.
In the following, $\| \tau_\alpha \|_\infty$ shall denote the maximal value
attained by $\tau_\alpha: \{ 1,\ldots,m \} \rightarrow [0,1]$ on $\{ 1,\ldots,m \}$.

\begin{lemma}
\label{lemma_invariance_delta}
Let $h$ stand for $h_u$ or $h_d$.
If $p^\ast \in [0,1]^m$, $\alpha^\ast>0$ with $p^\ast_{(i)} \neq \tau_{\alpha^\ast}(i)$
$\forall i \in \{ 1,\ldots,m \}$, then there exists $\delta>0$ such that
$p \in [0,1]^m$, $\tau_\alpha: \{ 1,\ldots,m \} \rightarrow [0,1]$ and
$\| p^\ast-p \| \vee \|\tau_{\alpha^\ast}-\tau_\alpha\|_\infty < \delta$
$\forall i \in \{ 1,\ldots,m \}$
imply $h(p,\alpha)=h(p^\ast,\alpha^\ast)$.
\end{lemma}

\begin{proof}
Let
$$\delta' = \min \left(
\left\{ \frac{p_i^\ast-p_{i-1}^\ast}{2}: p_{i-1}^\ast<p_i^\ast \right\}_{i=1,\ldots,m}
\cup
\{ |p_i^\ast - \tau_{\alpha^\ast}(i)| \}_{i=1,\ldots,m}
\right)$$
and let $\delta = \delta'/2$.

By assumption, $\| p-p^\ast \| < \delta < \delta'$, hence
$p_{i-1} < p_{i-1}^\ast + \delta' \leq p_i^\ast - \delta' < p_i$.
This means that $p_i^\ast$ and $p_i$ have the same ranks in $p^\ast$ and $p$,
respectively.

Moreover,
$|p_i^\ast - \tau_{\alpha^\ast}(i)| \leq
|p_i^\ast - \tau_\alpha(i)| + |\tau_\alpha(i) - \tau_{\alpha^\ast}(i)| \leq
|p_i^\ast - \tau_\alpha(i)| + \|\tau_\alpha-\tau_{\alpha^\ast}\|_\infty \leq
|p_i^\ast - \tau_\alpha(i)| + \delta$.
Hence $2 \delta = \delta' \leq |p_i^\ast - \tau_{\alpha^\ast}(i)| \leq
|p_i^\ast - \tau_\alpha(i)| + \delta$, meaning
that $\delta \leq |p_i^\ast - \tau_\alpha(i)|$ for $i \in \{ 1,\ldots,m \}$.

So $|p_i^\ast - p_i| < \delta \leq |p_i^\ast - \tau_\alpha(i)|$,
hence $p_i$ and $p_i^\ast$ lie on the same side of the testing threshold.
\end{proof}

\begin{proof}[Proof of Lemma \ref{lemma_stepupstepdown}]
1. The monotonicity of $h_u$ and $h_d$ follows from
Lemma \ref{lemma_invariance}
(part \ref{lemma_item_monotonicity}.) and
(part \ref{lemma_item_monotonicity_d}.), respectively.

2. To prove that $h_u$ satisfies the first part of
Condition \ref{condition_h},
it suffices to show that for $p,q \in [0,1]^m$, both
$q_i \leq p_i$ $\forall i \in h_u(p)$ and $q_i = p_i$ $\forall i \notin h_u(p)$
as well as
$q_i = p_i$ $\forall i \in h_u(p)$ and $q_i \geq p_i$ $\forall i \notin h_u(p)$
imply $h_u(p) = h_u(q)$.

Indeed, let $p,q \in [0,1]^m$ be such that
$q_i \leq p_i$ $\forall i \in h_u(p)$ and $q_i = p_i$ $\forall i \notin h_u(p)$.
We have $p_i \leq \tau_\alpha(|h_u(p)|)$ $\forall i \in h_u(p)$,
thus $q_i \leq p_i \leq \tau_\alpha(|h_u(p)|)$
$\forall i \in h_u(p)$ and $h_u(p) = h_u(q)$ by
Lemma \ref{lemma_invariance} (part \ref{lemma_item_invariance_rejection_area}.).

Similarly, let $p,q \in [0,1]^m$ be such that
$q_i = p_i$ $\forall i \in h_u(p)$ and $q_i \geq p_i$ $\forall i \notin h_u(p)$.
Using $p_i > \tau_\alpha(r_p(i))$ $\forall i \notin h_u(p)$,
it instantly follows that $q_i \geq p_i > \tau_\alpha(r_p(i))$ $\forall i \notin h_u(p)$
and thus $h_u(p) = h_u(q)$
by Lemma \ref{lemma_invariance} (part \ref{lemma_item_invariance_nonrejection_area}.).

To prove that $h_d$ satisfies the first part of Condition \ref{condition_h},
it equally suffices to show that for $p,q \in [0,1]^m$, both
$q_i \leq p_i$ $\forall i \in h_d(p)$ and $q_i = p_i$ $\forall i \notin h_d(p)$
as well as
$q_i = p_i$ $\forall i \in h_d(p)$ and $q_i \geq p_i$ $\forall i \notin h_d(p)$
imply $h_d(p) = h_d(q)$.

Indeed, let $p,q \in [0,1]^m$ be such that
$q_i \leq p_i$ $\forall i \in h_d(p)$ and $q_i = p_i$ $\forall i \notin h_d(p)$.
Using $p_i \leq \tau_\alpha(r_p(i))$ $\forall i \in h_d(p)$,
it immediately follows that $q_i \leq p_i \leq \tau_\alpha(r_p(i))$ $\forall i \in h_d(p)$
and thus $h_d(p) = h_d(q)$
by Lemma \ref{lemma_invariance} (part \ref{lemma_item_invariance_rejection_area_d}.).

Similarly, let $p,q \in [0,1]^m$ be such that
$q_i = p_i$ $\forall i \in h_d(p)$ and $q_i \geq p_i$ $\forall i \notin h_d(p)$.
We have $p_i > \tau_\alpha(|h_d(p)|+1)$ $\forall i \notin h_d(p)$,
thus $q_i \geq p_i > \tau_\alpha(|h_d(p)|+1)$
$\forall i \notin h_d(p)$ and $h_d(p) = h_d(q)$
by Lemma \ref{lemma_invariance} (part \ref{lemma_item_invariance_nonrejection_area_d}.).

3. As $\tau_\alpha(i)$ is continuous in $\alpha$ $\forall i \in \{ 1,\ldots,m \}$
by Condition \ref{condition_threshold},
for each $\epsilon_i>0$
there exists a $\delta_i>0$ such that $|\alpha^\ast-\alpha|<\delta_i$
implies $|\tau_{\alpha^\ast}(i)-\tau_{\alpha}(i)|<\epsilon_i$.
Applying continuity to $\epsilon_i=\delta$ yields a $\delta_i$
for each $i \in \{ 1,\ldots,m \}$,
where $\delta>0$ is given by Lemma \ref{lemma_invariance_delta}.
The second part of Condition \ref{condition_h} then follows for all
$p \in [0,1]^m$ and $\alpha \in [0,1]$ satisfying
$\| p-p^\ast \| \vee |\alpha-\alpha^\ast| < \min \{ \delta,\delta_1,\ldots,\delta_m \}$.
\end{proof}

\section{The Hommel procedure is not well-behaved}
\label{appendix_hommel}
The \cite{Hommel1988} procedure determines
the largest index $k$ satisfying $p_{(m-k+j)} > j \alpha/k$
for all $j = 1,\ldots,k$ and then rejects all the $H_{0i}$
with $p_i \leq \alpha/k$.
If no such $k$ exists, all hypotheses are rejected.

The \cite{Hommel1988} procedure $h(p,\alpha)$
is not a classical step-up or step-down procedure.
Given $p$, determining the index $k$ corresponds to
applying $m$ step-up procedures $h_j$, $j \in \{ 1,\ldots,m \}$,
to $P_j=( p_{(m-j+1)},\ldots,p_{(m)} )$ using the threshold
functions $\tau_j(i) = i\alpha/j$, where $i \in \{ 1,\ldots,j \}$.
Once $k_p=\max \{ j: h_j(P_j)=\emptyset \}$
is determined, rejections are calculated by applying
the \cite{Bonferroni1936} correction
(defined in Section \ref{subsection_procedures_examples})
at threshold $\alpha/k_p$
to all p-values $p$, i.e.\ $h(p,\alpha)=h_\text{Bonferroni}(p,m\alpha/k_p)$.

The \cite{Hommel1988} procedure satisfies the first and the third part of
Condition \ref{condition_h}.
However, for $q_i \geq p_i$ $\forall i \notin h(p,\alpha)$,
the second part of Condition \ref{condition_h} is not satisfied.
Consider $p=[\alpha/3+\epsilon, \alpha/2+\epsilon, 1]$,
where $0 < \alpha < 1$ and $0 < \epsilon \leq \alpha/6$.
Then $h_1(P_1)=\emptyset$, $h_2(P_2)=\emptyset$,
$h_3(P_3) = \{ 1,2 \}$, so $k_p=2$.
Therefore, $h(p,\alpha)=\{ 1 \}$.
Increasing $p_2$ to $p_2=2\alpha/3+\epsilon$ yields
$h_1(P_1)=h_2(P_2)=h_3(P_3)=\emptyset$, hence $k_p=3$ and $h(p,\alpha)=\emptyset$.

\end{document}